\documentclass[journal ]{new-aiaa}
\usepackage[utf8]{inputenc}
\usepackage{textcomp}

\usepackage{graphicx}
\usepackage{amsmath}

\usepackage[version=4]{mhchem}
\usepackage{siunitx}
\usepackage{longtable,tabularx}
\usepackage{caption}
\usepackage{booktabs}
\usepackage{amsthm}
\usepackage{float}
\usepackage{bm}
\usepackage{amssymb}  
\usepackage{array}
\usepackage{comment}
\usepackage{empheq}

\newcommand{\withtablespacing}[2]{%
  {%
  \renewcommand{\arraystretch}{#1}%
  #2%
  }%
}
\newcommand{\spaziatura}{0.8}
\newcommand{\matrice}[1]{\bm{\mathbf{#1}}}

\makeatletter
\newcommand{\belong}{\@ifnextchar\bgroup{\@belong@onearg}{\in \mathbb{R}^{3\times 3}}}
\newcommand{\@belong@onearg}[1]{\@ifnextchar\bgroup{\@belong@twoargs{#1}}{\in \mathbb{R}^{#1}}}
\newcommand{\@belong@twoargs}[2]{\in \mathbb{R}^{#1\times #2}}
\makeatother

\newcommand{\norm}[1]{\lVert #1 \lVert}
\newcommand{\Smat}{\bm{\mathcal{S}}}

\newcommand{\sgn}{\mathrm{sgn}}

\newcommand{\diag}{\mathrm{diag}}

\newcommand{\R}{\mathbb{R}}

\newcommand{\ie}{\textit{i.e.}}
\newcommand{\eg}{\textit{e.g.}}
\newcommand{\tenE}[1]{\cdot 10^{#1} }

\newtheorem{proposition}{Proposition}

\newtheorem{assumption}{Assumption}
\newtheorem{requirement}{Requirement}

\newcommand{\refcor}[1]{\hyperref[#1]{Corollary~\ref*{#1}}}
\newcommand{\reffig}[1]{\hyperref[#1]{Fig.~\ref*{#1}}}
\newcommand{\reftable}[1]{\hyperref[#1]{Table~\ref*{#1}}}
\newcommand{\refsec}[1]{\hyperref[#1]{Section~\ref*{#1}}}
\newcommand{\reftheorem}[1]{\hyperref[#1]{Theorem~\ref*{#1}}}
\newcommand{\refdefinition}[1]{\hyperref[#1]{Def.~\ref*{#1}}}
\newcommand{\refproposition}[1]{\hyperref[#1]{Prop.~\ref*{#1}}}
\newcommand{\refassumption}[1]{\hyperref[#1]{Ass.~\ref*{#1}}}
\newcommand{\refrequirement}[1]{\hyperref[#1]{Req.~\ref*{#1}}}

\newcommand{\sottosistema}[2]{ 

\begin{subequations} #2

  \begin{empheq}[left=\empheqlbrace]{align}
#1
  \end{empheq}
\end{subequations}

}

\usepackage{import}
\usepackage{tikz}
\usepackage{xcolor}
\usetikzlibrary{arrows.meta,calc,positioning}

\newcommand{\drawRotatedCircle}[6]{%
\begin{scope}
  \coordinate (C) at (#1,#2);
  \def\R{#3}
  \def\a{#4}
  \def\b{#5}
  \def\c{#6}

  \pgfmathsetmacro{\exx}{cos(\c)*cos(\b)}
  \pgfmathsetmacro{\exy}{sin(\c)*cos(\b)}

  \pgfmathsetmacro{\eyx}{-sin(\c)*cos(\a)+cos(\c)*sin(\b)*sin(\a)}
  \pgfmathsetmacro{\eyy}{ cos(\c)*cos(\a)+sin(\c)*sin(\b)*sin(\a)}

  \draw [
  color={rgb,255:red,0; green,29; blue,255},
  line width=1.5
]
    ($(C)+(\R*\exx,\R*\exy)$)
    arc[start angle=0, end angle=360,
        x radius=\R,
        y radius=\R,
        xslant=\eyx/\exx,
        yslant=\eyy/\exy];
\end{scope}
}

\newcommand{\ReferenceFrameFig}{
\begin{figure}[hbt]
\tikzset{every picture/.style={line width=0.75pt}} 
\centering
\begin{tikzpicture}[x=0.65pt,y=0.65pt,yscale=-1,xscale=1]

\drawRotatedCircle{349}{263}{83}{1}{11}{7}

\draw  [fill={rgb, 255:red, 255; green, 255; blue, 255 }  ,fill opacity=1 ] (339.33,121.32) .. controls (339.33,119.12) and (350.68,117.33) .. (364.67,117.33) -- (364.67,126) .. controls (350.68,126) and (339.33,127.78) .. (339.33,129.98) ;\draw  [fill={rgb, 255:red, 255; green, 255; blue, 255 }  ,fill opacity=1 ] (339.33,129.98) .. controls (339.33,131.62) and (345.65,133.03) .. (354.67,133.64) -- (354.67,136.28) -- (364.67,129.63) -- (354.67,122.34) -- (354.67,124.98) .. controls (345.65,124.37) and (339.33,122.96) .. (339.33,121.32)(339.33,129.98) -- (339.33,121.32) ;
\draw [color={rgb, 255:red, 208; green, 2; blue, 27 }  ,draw opacity=1 ][line width=1.5]    (409.04,222.71) -- (414.47,221.8) -- (481.83,241.54) ;
\draw [shift={(485.67,242.67)}, rotate = 196.34] [fill={rgb, 255:red, 208; green, 2; blue, 27 }  ,fill opacity=1 ][line width=0.08]  [draw opacity=0] (11.61,-5.58) -- (0,0) -- (11.61,5.58) -- cycle    ;
\draw [color={rgb, 255:red, 208; green, 2; blue, 27 }  ,draw opacity=1 ][line width=1.5]    (409.04,222.71) -- (385.08,286.26) ;
\draw [shift={(383.67,290)}, rotate = 290.66] [fill={rgb, 255:red, 208; green, 2; blue, 27 }  ,fill opacity=1 ][line width=0.08]  [draw opacity=0] (11.61,-5.58) -- (0,0) -- (11.61,5.58) -- cycle    ;
\draw [color={rgb, 255:red, 208; green, 2; blue, 27 }  ,draw opacity=1 ][line width=1.5]    (409.04,222.71) -- (443.34,174.59) ;
\draw [shift={(445.67,171.33)}, rotate = 125.49] [fill={rgb, 255:red, 208; green, 2; blue, 27 }  ,fill opacity=1 ][line width=0.08]  [draw opacity=0] (11.61,-5.58) -- (0,0) -- (11.61,5.58) -- cycle    ;
\draw [line width=1.5]    (349,264) -- (349,105.33) ;
\draw [shift={(349,101.33)}, rotate = 90] [fill={rgb, 255:red, 0; green, 0; blue, 0 }  ][line width=0.08]  [draw opacity=0] (11.61,-5.58) -- (0,0) -- (11.61,5.58) -- cycle    ;
\draw [line width=1.5]    (349,264) -- (504.2,264) ;
\draw [shift={(508.2,264)}, rotate = 180] [fill={rgb, 255:red, 0; green, 0; blue, 0 }  ][line width=0.08]  [draw opacity=0] (11.61,-5.58) -- (0,0) -- (11.61,5.58) -- cycle    ;
\draw [line width=1.5]    (349,264) -- (253.03,359.97) ;
\draw [shift={(250.2,362.8)}, rotate = 315] [fill={rgb, 255:red, 0; green, 0; blue, 0 }  ][line width=0.08]  [draw opacity=0] (11.61,-5.58) -- (0,0) -- (11.61,5.58) -- cycle    ;
\draw [color={rgb, 255:red, 65; green, 117; blue, 5 }  ,draw opacity=1 ][fill={rgb, 255:red, 65; green, 117; blue, 5 }  ,fill opacity=1 ][line width=1.5]    (349.8,264) -- (385.08,239.13) -- (416.13,217.23) -- (499.93,158.15) ;
\draw [shift={(503.2,155.85)}, rotate = 144.82] [fill={rgb, 255:red, 65; green, 117; blue, 5 }  ,fill opacity=1 ][line width=0.08]  [draw opacity=0] (11.61,-5.58) -- (0,0) -- (11.61,5.58) -- cycle    ;
\draw [color={rgb, 255:red, 65; green, 117; blue, 5 }  ,draw opacity=1 ][fill={rgb, 255:red, 65; green, 117; blue, 5 }  ,fill opacity=1 ][line width=1.5]    (409.64,220.05) -- (362.97,167.83) ;
\draw [shift={(360.3,164.85)}, rotate = 48.21] [fill={rgb, 255:red, 65; green, 117; blue, 5 }  ,fill opacity=1 ][line width=0.08]  [draw opacity=0] (11.61,-5.58) -- (0,0) -- (11.61,5.58) -- cycle    ;
\draw [color={rgb, 255:red, 65; green, 117; blue, 5 }  ,draw opacity=1 ][fill={rgb, 255:red, 65; green, 117; blue, 5 }  ,fill opacity=1 ][line width=1.5]    (408.64,220.55) -- (438.79,125.21) ;
\draw [shift={(440,121.4)}, rotate = 107.55] [fill={rgb, 255:red, 65; green, 117; blue, 5 }  ,fill opacity=1 ][line width=0.08]  [draw opacity=0] (11.61,-5.58) -- (0,0) -- (11.61,5.58) -- cycle    ;
\draw  [fill={rgb, 255:red, 0; green, 0; blue, 0 }  ,fill opacity=1 ] (404.57,221.8) .. controls (404.57,219.07) and (406.79,216.85) .. (409.52,216.85) .. controls (412.26,216.85) and (414.47,219.07) .. (414.47,221.8) .. controls (414.47,224.53) and (412.26,226.75) .. (409.52,226.75) .. controls (406.79,226.75) and (404.57,224.53) .. (404.57,221.8) -- cycle ;
\draw  [color={rgb, 255:red, 245; green, 166; blue, 35 }  ,draw opacity=1 ][fill={rgb, 255:red, 255; green, 116; blue, 0 }  ,fill opacity=1 ] (391.8,205.85) .. controls (391.8,203.12) and (394.02,200.9) .. (396.75,200.9) .. controls (399.48,200.9) and (401.7,203.12) .. (401.7,205.85) .. controls (401.7,208.58) and (399.48,210.8) .. (396.75,210.8) .. controls (394.02,210.8) and (391.8,208.58) .. (391.8,205.85) -- cycle ;

\draw [color={rgb, 255:red, 189; green, 16; blue, 224 }  ,draw opacity=1 ]   (365.2,75.7) -- (406.61,118.76) ;
\draw [shift={(408,120.2)}, rotate = 226.12] [color={rgb, 255:red, 189; green, 16; blue, 224 }  ,draw opacity=1 ][line width=0.75]    (10.93,-3.29) .. controls (6.95,-1.4) and (3.31,-0.3) .. (0,0) .. controls (3.31,0.3) and (6.95,1.4) .. (10.93,3.29)   ;
\draw [color={rgb, 255:red, 189; green, 16; blue, 224 }  ,draw opacity=1 ]   (373.2,69.7) -- (414.61,112.76) ;
\draw [shift={(416,114.2)}, rotate = 226.12] [color={rgb, 255:red, 189; green, 16; blue, 224 }  ,draw opacity=1 ][line width=0.75]    (10.93,-3.29) .. controls (6.95,-1.4) and (3.31,-0.3) .. (0,0) .. controls (3.31,0.3) and (6.95,1.4) .. (10.93,3.29)   ;
\draw [color={rgb, 255:red, 189; green, 16; blue, 224 }  ,draw opacity=1 ]   (381.2,62.7) -- (422.61,105.76) ;
\draw [shift={(424,107.2)}, rotate = 226.12] [color={rgb, 255:red, 189; green, 16; blue, 224 }  ,draw opacity=1 ][line width=0.75]    (10.93,-3.29) .. controls (6.95,-1.4) and (3.31,-0.3) .. (0,0) .. controls (3.31,0.3) and (6.95,1.4) .. (10.93,3.29)   ;

\draw (313,155.33) node [anchor=north west][inner sep=0.75pt]    {$\mathcal{F}_{i}$};
\draw (475.83,129) node [anchor=north west][inner sep=0.75pt]  [color={rgb, 255:red, 65; green, 117; blue, 5 }  ,opacity=1 ]  {$\mathcal{F}_{h}$};
\draw (508.33,208.17) node [anchor=north west][inner sep=0.75pt]  [color={rgb, 255:red, 208; green, 2; blue, 27 }  ,opacity=1 ]  {$\mathcal{F}_{b}$};
\draw (404.3,231.4) node [anchor=north west][inner sep=0.75pt]    {$f$};
\draw (391.3,176.4) node [anchor=north west][inner sep=0.75pt]  [color={rgb, 255:red, 255; green, 116; blue, 0 }  ,opacity=1 ]  {$l$};
\draw (269.8,352.9) node [anchor=north west][inner sep=0.75pt]    {$x^{i}$};
\draw (493.8,272) node [anchor=north west][inner sep=0.75pt]    {$y^{i}$};
\draw (358.3,87) node [anchor=north west][inner sep=0.75pt]    {$z^{i}$};
\draw (508.3,138.17) node [anchor=north west][inner sep=0.75pt]  [color={rgb, 255:red, 65; green, 117; blue, 5 }  ,opacity=1 ]  {$x^{h}$, $\bm{e}_r$};
\draw (364.3,143.67) node [anchor=north west][inner sep=0.75pt]  [color={rgb, 255:red, 65; green, 117; blue, 5 }  ,opacity=1 ]  {$y^{h}$, $\bm{e}_\theta$};
\draw (444,104.6) node [anchor=north west][inner sep=0.75pt]  [color={rgb, 255:red, 65; green, 117; blue, 5 }  ,opacity=1 ]  {$z^{h}$, $\bm{e}_h$};
\draw (448.47,162.57) node [anchor=north west][inner sep=0.75pt]  [color={rgb, 255:red, 208; green, 2; blue, 27 }  ,opacity=1 ]  {$z^{b}$};
\draw (391.13,279.23) node [anchor=north west][inner sep=0.75pt]  [color={rgb, 255:red, 208; green, 2; blue, 27 }  ,opacity=1 ]  {$x^{b}$};
\draw (486.47,233.23) node [anchor=north west][inner sep=0.75pt]  [color={rgb, 255:red, 208; green, 2; blue, 27 }  ,opacity=1 ]  {$y^{b}$};
\draw (314.33,116.67) node [anchor=north west][inner sep=0.75pt]    {$\bm{\omega }_{\oplus }^{i}$};
\draw (387.6,367) node   [align=left] {
\textcolor[rgb]{0,0.11,1}{Follower Orbit}};

\draw (287.33,66.67) node [anchor=north west,inner sep=0.75pt] {%
  {\color[rgb]{0.74,0.06,0.88}%
   $\Smat\!(\bm{\omega}_{\oplus}^{\,i})\bm{r}_{f}^{\,i}$}%
};

\begin{scope}
  \def\Rearth{14pt}

  \clip (349,264) circle (\Rearth);

  \node at (349,264) {\includegraphics[width=36pt]{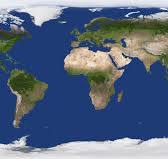}};
\end{scope}

\draw[black, line width=0.6pt] (349,264) circle (14pt);

\end{tikzpicture}    
\caption{ECI and Hill reference frames.}
\label{Frames}
\end{figure}
}

\usepackage{xstring}

\newcommand{\doi}[1]{%
\IfSubStr{#1}{doi.org/}
  {\StrBehind{#1}{doi.org/}[\cleandoi]}
  {\def\cleandoi{#1}}%
DOI: \href{https://doi.org/\cleandoi}{\cleandoi}%
}

\title{Leader-Follower Formation Control Using Differential Drag and Effective Surface Regulation\footnote{This work was presented at the 2025 AAS/AIAA Astrodynamics Specialist Conference, Boston, Massachusetts, August 10-14 2025. Paper number: AAS 25-661 \cite{bocci2025AstrodynSpecCOnf}. }}

\author{Alessio Bocci\footnote{PhD Candidate,  Department of Electrical Engineering, Science and Technology, alessio.bocci@uit.no (corresponding author).}   , Jos\'{e} Juan Corona-S\'{a}nchez\footnote{Associate Professor, Department of Electrical Engineering, Science and Technology, jose.j.sanchez@uit.no.}, and Raymond Kristiansen\footnote{Professor, Department of Electrical Engineering, Science and Technology, raymond.kristiansen@uit.no.}}
\affil{ UiT - The Arctic University of Norway, 8514 Narvik, Norway}

\begin{document}

\maketitle

\section*{Nomenclature}

{\renewcommand\arraystretch{1.0}
\noindent\begin{longtable*}{@{}l @{\quad=\quad} l@{}}

$\bm{a}_D$, $C_D$ & drag acceleration $[$m/s$^2$$]$ and drag coefficient $[-]$\\
$\bm{a}_P$ & perturbative accelerations' vector $[$m/s$^2$$]$\\
$\mathcal{F}_b$, $\mathcal{F}_h$, $\mathcal{F}_i$ & spacecraft body frame, Hill frame, and Earth-centered inertial frame  \\
$\bm{g}$, $\matrice{Q}$, $\mathfrak{R}$ & LQR control gain and costs\\
$\mathscr{i}_f$ & orbit inclination $[$rad$]$\\
$\matrice{J}$, $M$ & inertia tensor $[$kg$\cdot$m$^2$$]$ and mass $[$kg$]$\\
$k_s,\, \matrice{K}_\omega$ & surface gain $[$1/s$]$ and angular velocity gain $[$N$\cdot$m$\cdot$s$]$ \\
$\bm{m}$, $\hat{\bm{m}}$  & $\sgn(\bm{\eta})$  and value at $t=\hat{t}$ $[-]$\\
$\bm{n}_j$ & normal vector of the jth spacecraft surface $[-]$\\  
$\bm{q}$ & unit quaternion $[-]$\\
$\matrice{R}_i^b$, $\matrice{R}_i^h$ & $\mathcal{F}_i$ to $\mathcal{F}_b$ rotation matrix and $\mathcal{F}_i$ to $\mathcal{F}_h$ rotation matrix\\
$\bm{r}^i$, $\bm{v}^i$  & spacecraft position vector $[$m$]$ and velocity vector $[$m/s$]$ \\
$\Smat(\cdot)$ & cross product matrix \\
$S$, $\Delta S$ & effective surface and leader-follower effective surface difference $[$m$^2$$]$\\
$\bm{s}$, $\matrice{\sigma}$ & $[\hat{S}_1,\hat{S}_2,\hat{S}_3]$ spacecraft nominal surfaces vector and matrix $[$m$^2$$]$ \\
$\bm{s}_m$ & $\bm{\sigma}\hat{\bm{m}}$ signed surface vector at $t=\hat{t}$ $[$m$^2$$]$\\
$\hat{t}$ & beginning instant of control operations $[$s$]$\\
$T$, $\nu$ & orbital period $[$s$]$ and angular velocity $[$rad/s$]$ \\
$\bm{u}^b$ & control torque $[$N$\cdot$m$]$\\
$\bm{w}^i$ & atmosphere-spacecraft relative velocity $[$m/s$]$\\
$\bm{x}^h$ & leader-follower relative state \\

$\Gamma_s$ & admissible effective surface range $[$m$^2$$]$\\
$\bm{\Gamma}_\omega$ & set of constraints on $\bm{\omega}^b$\\
$\delta\bm{r}^h$ & leader-follower relative position vector in $\mathcal{F}_h$ $[$m$]$\\
$\lambda_\mathscr{i}$ & coefficient of the Hill-Clohessy-Wiltshire equations $[$1/s$^2$$]$ \\
$\mu$, $\bm{\omega}_\oplus$ & Earth gravitational constant $[$m$^3$/s$^2 $$]$ and angular velocity $[$rad/s$]$  \\
$\bm{\xi}^i$, $\bm{\eta}^b$ & $\bm{w}^i/\norm{\bm{w}^i}$ and its projection onto $\mathcal{F}_b$ $ [-]$\\
$\rho$, $\beta$ & atmospheric density $[$kg/m$^3$$]$ and $\rho C_D/(2M)$ $[$1/m$^3$$]$\\
 $\phi$, $\bm{\psi}^b$ &$\bm{s}_m^\top   \matrice{R}_i^b \dot{\bm{\xi}}^i$  $[$m$^2$/s$]$ and $- \bm{\mathcal{S}} (\bm{\eta}^b)\bm{s}_m$  $[$m$^2$$]$ (coefficients of the effective surface dynamics) \\
$\bm{\omega}^b$ & spacecraft angular velocity $[$rad/s$]$ \\

\multicolumn{2}{@{}l}{Subscripts}\\
$d$ & desired \\
$e$&error\\
$f$ & follower\\
$l$ & leader\\

\multicolumn{2}{@{}l}{Superscripts}\\
$b$ & quantity expressed in $\mathcal{F}_b$\\
$h$ & quantity expressed in $\mathcal{F}_h$\\
$i$ & quantity expressed in $\mathcal{F}_i$\\

\end{longtable*}}
\setcounter{table}{0}

\section{Introduction}

The modern age features significant space activity, with many countries prioritizing space exploration. CubeSats lower mission costs, but their lack of propulsion poses challenges for orbital control (OC). Standard CubeSat OC exploits atmospheric drag, whereas constellations or formations employ differential drag (DD). In leader-follower formations, the "leader" (or "chief" or "target") is the main spacecraft, and the "follower" (or "deputy" or "chaser") operates relative to it. 

Research on DD-based OC began in 1989 with Leonard et al. \cite{leonard1989orbital}, who aimed to provide an alternative to the traditional chemical-thruster-based methodology \cite{vassar1985formationkeeping,redding1989linear}. The controller, actuated by tiltable drag panels, regulates the follower’s mean position relative to the leader while incorporating an eccentricity-minimization scheme. Kumar et al. \cite{kumar2011differential} explored the altitude range for DD formation control on a circular Sun-synchronous orbit. They concluded that altitudes below 400 km are optimal and proposed a proportional-integral-derivative (PID) relative position control.  

Bevilacqua and Romano \cite{bevilacqua2008rendezvous} designed a DD-based OC for multiple-spacecraft rendezvous maneuvers actuated by retractable drag plates. Based on the Schweighart-Sedwick equations  \cite{schweighart2002high}, the algorithm drives the chasers to a closed orbit around the target and then acts on its eccentricity to nullify the semimajor axis. In  \cite{bevilacqua2010multiple}, the authors improved the design by allowing the chasers' ballistic coefficients to vary while keeping the target's constant.

Horsley et al. \cite{horsley2011investigation,horsley2013small} explored a generic aerodynamic force control approach. Indeed, DD only allows control of the relative in-plane motion, whereas differential lift can control the out-of-plane motion. However, the lift is significantly weaker than the drag, so this method is only beneficial at low relative inclinations and altitudes. Similarly, Shao et al. \cite{shao2015satellite} investigated aerodynamic-force-based control for rendezvous maneuvers. Their controller significantly reduces the relative position error and is actuated by tiltable drag panels.

 Omar and Wersinger \cite{omar2015satellite} considered two identical CubeSats in formation flight on the same orbit and designed a DD-based OC that switches between minimum and maximum drag configurations and that requires only a 3-axis attitude control (feasible for many CubeSats). However, intermediate configurations are also exploitable, as shown by Smith et al. \cite{smith2017fast}, who investigated the minimum-time establishment of CubeSat constellations.

Ben-Yaacov and Gurfil \cite{ben2013long} designed a nonlinear DD-based OC for long-term satellite cluster flight using relative mean orbital elements. They showed that the relative semimajor axis can be regulated by modulating the cross-sectional area, whereas the relative eccentricity is uncontrollable for quasi-circular orbits. In \cite{benyaacov2014stability}, they improved the design by introducing the differential Brouwer-Lyddane mean elements, investigating the possibility of controlling the differential mean inclination leveraging the out-of-plane component of the DD force, and developing a nonlinear DD-based collision avoidance
controller. Further analysis of the controller's robustness was provided in \cite{ben2016covariance}.

Differential drag control is inherently an orbit-attitude-coupled problem. Pastorelli et al. \cite{pastorelli2015differential} introduced the virtual thrusters (\ie, the drag sails) concept and proposed a Lyapunov controller with on/off actuators logic, simultaneously controlling both
positional and rotational motion using drag only. Similarly, Sun et al. \cite{sun2017roto} developed an orbit-attitude controller using rotatable and extendable flat plates as the sole actuators for generating aerodynamic forces and torques.

 Ivanov et al. \cite{ivanov2018study} considered two spherical spacecraft in formation flight with solar panels and proposed a DD force model parametrized using two unit vectors: one normal to the solar panels and the other aligned with the incoming airflow velocity. They expressed the DD force as a function of four attitude angles, two latitudes, and two longitudes. Their algorithm first evaluates the linear-quadratic regulator (LQR) optimal orbit control and then solves nonlinear equations to obtain the attitude angles. A drawback of the method is that the existence of a solution to the nonlinear system depends on the magnitude of the LQR output. Moreover, it can be computationally expensive.

Harris et al. \cite{harris2020linear} derived a DD-based orbital-attitude controller. The relative orbital motion was modeled using the Hill-Clohessy-Wiltshire (HCW) plus drag formulation of Silva \cite{silva2008formulation}, and the attitude was parametrized using the modified Rodrigues parameters. The coupled system was reduced to a linear time-invariant (LTI) system with the attitude vector as input, pre-multiplied by a sensitivity matrix that accounts for the deputy ballistic coefficient's sensitivity to arbitrary attitude variations. The controller was designed using the LQR approach.

The reviewed literature reveals a clear gap: there is no explicit, non-linearized algorithm for managing the spacecraft's surface that interacts with the atmosphere-spacecraft relative velocity solely through attitude maneuvers, without using drag panels. We denote this surface as the effective surface (ES), where the adjective "effective" means the only surface capable of generating a net drag force per unit of mass to vary the orbital motion. In \cite{bocciIAC2024}, we 
obtained the required ES profile for a virtual leader and real follower formation flight control via LQR. 
 The adjective "virtual" indicates that the leader is not actually orbiting; the onboard computer instantaneously simulates its trajectory. This assumption relates to the QBDebris mission \cite{ellingsencubesats,bocciIAC2024}, organized by UIT, the Arctic University of Norway, to which this work is closely related. 

\subsection{Objective and novelties}

The objective of this work is to design a reliable coupled orbit-attitude DD-based controller for CubeSat formation flight that, starting from any initial orbital conditions consistent with the HCW model assumptions and from arbitrary initial attitude conditions, asymptotically drives the inter-satellite distance to a prescribed value. Three main novelties are present in the paper. 

\begin{enumerate}
    \item The definition of the ES and its dynamics using a different mathematical formalism that provides a better physical understanding and a straightforward setup for proving the asymptotic stability via Lyapunov theory.
    \item The development of a controller that directly acts on angular velocity. This approach differs from traditional reference-attitude tracking methods. These may be computationally expensive because they require precomputing attitude and angular velocity profiles by solving nonlinear equations. In contrast, our algorithm evaluates the optimal angular velocity to achieve the desired ES, with the attitude following suit.
    \item We successfully showed that the angular velocity approach may enable the simultaneous meeting of some, although coarse, additional attitude constraints beyond just the leader-follower relative positioning.
\end{enumerate}

\section{Preliminaries} \label{preliminaries}
This section introduces the reference frames and the fundamental equations used in this work. Referring to \reffig{Frames}, we consider three right-handed Cartesian reference frames: the Earth-Centered Inertial (ECI) frame $\mathcal{F}_i$, the Hill frame $\mathcal{F}_h$, and the spacecraft body frame $\mathcal{F}_b$, where the superscripts $i$, $h$, and $b$ denote the corresponding frames. $\mathcal{F}_i$ is a non-rotating frame centered on the Earth. The $x^i$-axis is in the vernal equinox direction, $x^i-y^i$ is the equatorial plane, and the $z^i$-axis coincides with the Earth's rotation axis pointing northward. Consistently with the convention adopted in \cite{riano2020differential,maestrini2023relative}, the frame $\mathcal{F}_h$ is defined with origin at the follower, with the $(x^h,y^h,z^h)$ axes aligned with  
\begin{align}
    \bm{e}_r=\bm{r}_f^i\norm{\bm{r}_f^i}^{-1}, \quad \bm{e}_h=\Smat(\bm{r}_f^i)\bm{v}_f^i\norm{\Smat(\bm{r}_f^i)\bm{v}_f^i}^{-1}, \quad \bm{e}_\theta=\Smat(\bm{e}_h)\bm{e}_r
\end{align}

\noindent
where $\bm{r}_f^i$ is the follower position vector, $\bm{v}_f^i$ its velocity vector, and,  given any two generic vectors $\bm{x}$ and $\bm{y}$, the skew-symmetric cross-product matrix $\Smat$ is such that $\bm{x}\times\bm{y}=\Smat(\bm{x})\bm{y}$. We define $\Smat\belong$ and $\matrice{\Theta}\belong{4}{4}$ as
\withtablespacing{0.55}{
\begin{align}
    \Smat (\bm{x})=\left[\begin{array}{ccc}
        0 & -x_3 &x_2  \\
         x_3&0&-x_1\\
         -x_2&x_1&0
    \end{array}\right], \quad   \bm{\Theta}(\bm{x})=\left[\begin{array}{cc}
        0 & -\bm{x}^\top  \\
         \bm{x}& -\bm{\mathcal{S}}(\bm{x}) 
    \end{array}\right]
\end{align}
}
\noindent
$\mathcal{F}_h$ rotates relatively to $\mathcal{F}_i$ with angular velocity $\bm{\nu}_f^i=\Smat(\bm{r}_f^i)\bm{v}_f^i/\norm{\bm{r}_f^i}^2$ and the rotation matrix $\matrice{R}_i^h=\left[\bm{e}_r,\bm{e}_{\theta},\bm{e}_h\right]^\top\in \text{SO}(3)$ is such that  $\bm{x}^h=\matrice{R}_i^h\bm{x}^i$.  Finally, $\mathcal{F}_b$ is centered on the follower's center of mass, and the axes coincide with the principal inertia axes. $\mathcal{F}_b$ rotates relatively to $\mathcal{F}_i$ with angular velocity $\bm{\omega}^b$ and the rotation matrix  $\matrice{R}_i^b\in \text{SO}(3)$ is such that
$\bm{x}^b =\matrice{R}_i^b\bm{x}^i$. For a spacecraft orbiting in Low Earth Orbit (LEO), the two-body problem with atmospheric drag  is  
\begin{align}
    \ddot{\bm{r}}^i=-\mu\norm{\bm{r}^i}^{-3}\bm{r}^i-\beta S \norm{\bm{w}^i}\bm{w}^i+\bm{a}_P^i
    \label{2body}
\end{align}

\ReferenceFrameFig

\noindent
with $\beta$ and the atmosphere-spacecraft relative velocity vector $\bm{w}^i$ defined as
\begin{align}
    \beta(M,\rho,C_D)=\rho C_D/(2M), \quad \bm{w}^i=\bm{v^i}-\bm{\mathcal{S}} (\bm{\omega}_{\bm{\oplus}}^i)\bm{r}^i
    \label{betacoeffient}
\end{align}

\noindent
where $\mu=398600$ km$^3$/s$^2$ is the Earth's gravitational constant, $S$ is the ES, $C_D$ is the drag coefficient, $M$ is the mass,  $\rho$ is the atmospheric density, $\bm{\omega}_{\bm{\oplus}}^i=[0,0,7.292115486 \cdot 10^{-5}]^\top$ rad/s is the Earth's  angular velocity, and $\bm{a}_p^i$ is a vector that includes additional perturbing accelerations. Let $\matrice{I}_{n\times n}$ be the $n\times n$ identity matrix, we define $\bm{\xi}^i$ and $\dot{\bm{\xi}}^i$ as 
\begin{align}
            \bm{\xi}^i=\bm{w}^i/\norm{\bm{w}^i}, \quad\dot{\bm{\xi}}^i&=\norm{\bm{w}^i}^{-1}[\bm{\xi}^i(\bm{\xi}^i)^\top-\matrice{I}_{3\times 3}]\left[\Smat (\bm{\omega}^i_{\bm{\oplus}})\bm{v}^i+\mu\norm{\bm{r}^i}^{-3}\bm{r}^i
        \right]
        \label{csidotfor}
    \end{align}

\noindent

\section{Effective Surface Definition}
 The ES is the spacecraft's cross-sectional area projected along $\bm{\xi}^i$, and its definition relates to the drag force, \ie, a consequence of the atmosphere's molecules-surface interaction \cite{canuto2018spacecraft}. For non-trivial shapes, the ES can be computed using advanced techniques, \eg, convex polygons \cite{ben2015analytical}. However, a major drawback of these sophisticated algorithms is the difficulty of explicitly expressing ES dynamics in terms of attitude parameters, a key step in controller design.  

The ES can be parametrized in several ways, \eg, as a function of the incidence angle \cite{du1996using} or through auxiliary angles \cite{ivanov2018study}. For a cuboid with $N$ flat faces, the traditional definition of the ES is \cite{markley2014correction}
\begin{align}
    S =\sum_{j=1}^N (\bm{\eta}^{b})^\top \bm{n}_j^b \hat{S}_j, \quad \bm{\eta}^b=\matrice{R}_i^b\bm{\xi}^i
    \label{Seff}
\end{align}

\noindent
where $\bm{n}_j^b$ is the unit vector normal to the jth surface $\hat{S}_j$. For a CubeSat, $N=6$ and the usual setup is shown in \reftable{tabnormal}. If $(\bm{\eta}^{b})^\top \bm{n}_j^b<0$, $\hat{S}_j$ is not exposed to the flow, and the jth term must be omitted in the sum. Thus, we can rearrange \eqref{Seff} as
\begin{align}
    S =\frac{1}{2}\sum_{j=1}^6 c_j \hat{S}_j=\frac{1}{2}\sum_{j=1}^3 c_j \hat{S}_j+\frac{1}{2}\sum_{j=4}^6 c_j \hat{S}_j=\frac{1}{2}\mathfrak{S}_1+\frac{1}{2}\mathfrak{S}_2, \quad c_j=(\bm{\eta}^{b})^\top \bm{n}_j^b [1+\sgn((\bm{\eta}^{b})^\top \bm{n}_j^b )]
    \label{Seff2}
\end{align}

\noindent
with $\sgn(\bm{\eta}^b)=[\sgn(\eta^b_1),\sgn(\eta^b_2),\sgn(\eta^b_3)]^\top$, $\sgn(\eta^b_j)=2H(\eta^b_j)-1$, and $H(\cdot)$ the  Heaviside step function. Moreover, $\mathfrak{S}_1=\bm{s}^\top \bm{\eta}^b+\bm{m}^\top \bm{\sigma}  \bm{\eta}^b$ and $\mathfrak{S}_2=-\bm{s}^\top \bm{\eta}^b+\bm{m}^\top \bm{\sigma}  \bm{\eta}^b$ with the surface vector $\bm{s}$, the surface tensor $\matrice{\sigma}$, and $\bm{m}$ defined as
\begin{align}
\bm{s}=[\hat{S}_1,\hat{S}_2,\hat{S}_3]^\top, \quad
 \bm{\sigma} =\diag(\bm{s}), \quad   \bm{m}=\sgn( \bm{\eta}^b)
\end{align}

\noindent
Hence, inserting the values of $\mathfrak{S}_1$ and $\mathfrak{S}_2$ into \eqref{Seff2}, the final expression of the ES becomes
\begin{align}
    S =\bm{m}^\top \bm{\sigma}   \bm{\eta}^b 
    \label{SeffFinale}
\end{align}

Starting from the $\bm{\eta}^b$-definition in \eqref{Seff}, $\dot{\bm{\eta}}^b=\Smat (\bm{\eta}^b) \bm{\omega}^b+ \matrice{R}_i^b\dot{\bm{\xi}^i}$, where we used the rotation matrix kinematic equation 
\begin{align}
    \dot{\matrice{R}}_i^b=-\Smat(\bm{\omega}^b)\matrice{R}_i^b
    \label{kinematicequation}
\end{align}

\noindent
Since $\mathrm{d}H(\eta_j^b)/\mathrm{d}\eta_j^b=\delta(\eta_j^b)$ where $\delta(\eta_j^b)\colon \eta_j^b\delta(\eta_j^b)=0$ is the Dirac Delta function \cite{burrows1990fourier}, the time derivative of \eqref{SeffFinale} is
 \begin{align}
    \dot{S} =[2(\bm{\eta}^b)^\top \diag(\delta(\bm{\eta}^b))+\bm{m}^\top ]\bm{\sigma}  [\Smat (\bm{\eta}^b) \bm{\omega}^b+ \matrice{R}_i^b\dot{\bm{\xi}^i}]=\bm{m}^\top \bm{\sigma}  [\Smat (\bm{\eta}^b) \bm{\omega}^b+ \matrice{R}_i^b\dot{\bm{\xi}^i}]
        \label{Seffdot}
\end{align}

At most, three surfaces can simultaneously face $\bm{\eta}^b$, resulting in the eight configurations listed in \reftable{tabsurfpos}. Equation \eqref{SeffFinale} is equivalent to $S=\norm{\matrice{\sigma}\bm{\eta}^b}_1$. Thus, the ES is a continuous function but non-differentiable whenever any $\eta_j^b=0$, as illustrated in \reffig{configtrans}. Consequently, $\dot{S}$ exhibits discontinuities at configuration transition points, which preclude the application of Lyapunov theory for controller design and stability proofs. This fact leaves two possibilities:
\begin{enumerate}
    \item assuming that ES control is performed while remaining within the same configuration, where $S$ is a $\mathcal{C}^2$ function
    \item considering the problem as a hybrid system and using the related theory (\eg, switching theory \cite{schlanbusch2012hybrid, schlanbusch2014hybrid}) to prove asymptotic stability in the case of configuration transitioning. 
\end{enumerate}
 In this work, we choose the first option and, thus, we state the following configuration-belonging requirement:

\begin{figure}[hbt!]
\centering
\begin{minipage}[l]{0.576\textwidth}
  \centering
  \includegraphics[width=\textwidth, trim=0.1cm 0.2cm 0.7cm 0.2cm, clip]{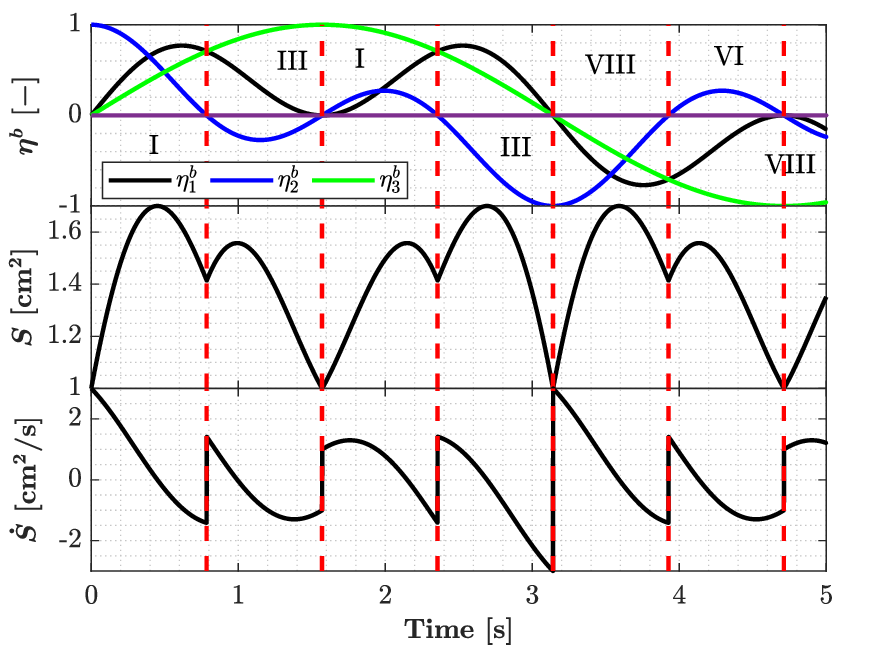}
  \caption{\centering Configuration transitions.}
  \label{configtrans}
\end{minipage}\hfill
\begin{minipage}[r]{0.40\textwidth}
    \centering
\withtablespacing{\spaziatura}{
\captionof{table}{CubeSats surfaces and normals}
  \label{tabnormal}

\begin{tabular}{lll}
\toprule
Face $j$& Surface $j$ &Unit vector $\bm{n}_j^b$ \\     
\midrule
 1 & $S_1$&$[1,0,0]^\top $ \\ 
         2 & $S_2$&$[0,1,0]^\top $ \\ 
         3 & $S_3$&$[0,0,1]^\top $ \\ 
         4 & $S_4=S_1$&$[-1,0,0]^\top $ \\ 
         5 & $S_5=S_2$&$[0,-1,0]^\top $ \\
         6 & $S_6=S_3$&$[0,0,-1]^\top $ \\
\bottomrule
\end{tabular}
}

 \vspace{0.5cm} 

\withtablespacing{\spaziatura}{
  \captionof{table}{ Possible configurations}
  \label{tabsurfpos}

\begin{tabular}{lll}
\toprule

Config.& Surfaces &$\bm{m}=\sgn(\bm{\eta}^b)$ \\
\midrule
      I &  1, 2, 3 & $[1,1,1]^\top$ \\
II& 1, 2, 6 & $[1,1,-1]^\top$\\
 III&1, 5,  3 & $[1,-1,1]^\top$ \\
 IV&1, 5,  6  & $[1,-1,-1]^\top$ \\
 V&4, 2, 3 & $[-1,1,1]^\top$ \\
 VI&4, 2,  6 & $[-1,1,-1]^\top$\\
VII &4, 5, 3  & $[-1,-1,1]^\top$\\
 VIII&4, 5,  6 & $[-1,-1,-1]^\top$\\
\bottomrule
\end{tabular}
}

\end{minipage}
\end{figure}

\begin{requirement} \label{RequirementI}
If the spacecraft at $t=\hat{t}\geq 0$ is in a  configuration that corresponds to $\hat{\bm{m}}=\sgn(\bm{\eta}^b)|_{t=\hat{t}}$, then, $\forall t\geq \hat{t}$, the rotation profile $\matrice{R}_i^b(\bm{\omega}^b)$ must be  such that $\bm{m}=\hat{\bm{m}} \, \forall t\geq \hat{t}$, which implies 
\begin{equation}
\diag(\hat{\bm{m}})\bm{\eta}^b=\diag(\hat{\bm{m}})\matrice{R}_i^b(\bm{\omega}^b)\bm{\xi}^i\geq 0, \quad \forall t\geq \hat{t}\geq 0    
\label{conficonstr}
\end{equation}

\noindent
where the notation $\matrice{R}_i^b(\bm{\omega}^b)$ means that $\matrice{R}_i^b$ is a function of $\bm{\omega}^b$ and not a multiplication.
\end{requirement}

If \refrequirement{RequirementI} is satisfied, the ES can be rewritten as $S=\bm{s}_m^\top \bm{\eta}^b$, where we denote $\bm{s}_m=\matrice{\sigma} \hat{\bm{m}}$ as the signed surface vector. Then, the ES dynamics \eqref{Seffdot} reduces to
\begin{align}
    \dot{S}&=(\bm{\psi}^{b})^{\top}\bm{\omega}^b+\phi, \quad \phi=\bm{s}_m^\top   \matrice{R}_i^b \dot{\bm{\xi}}^i,\quad \bm{\psi}^{b}=- \Smat (\bm{\eta}^b)\bm{s}_m
    \label{Sdotsop}
\end{align}

 To the best of our knowledge, \eqref{Sdotsop} offers a novel ES dynamics formulation, showing that $\dot{S}$ varies linearly with $\bm{\omega}^b$, scaled by the time-varying $\bm{\psi}^{b}$ and affected by a disturbance $\phi$ linked to the orbital state via $\dot{\bm{\xi}}^i$ in \eqref{csidotfor}.
Moreover, note that $\min(\bm{s}) \leq S\leq \norm{\bm{s}}\cdot\norm{\bm{\eta}^b}=\norm{\bm{s}}$ and  $S=\norm{\bm{s}}$ if $\bm{\eta}^b$ is perfectly aligned with the diagonal towards one corner of the CubeSat and equal to the value $\hat{\bm{\eta}}^b=\bm{s}_m/\norm{\bm{s}}$. This remark leads to the following admissibility requirement:

\begin{requirement} \label{reqSd}
    Any ES profile must be admissible, i.e., $S\in\Gamma_s=[\min(\bm{s}),\norm{\bm{s}}]\subset\R^+\, \forall t\geq 0$ with $\R^+=\{x\in\R \colon x\geq0\}$.
\end{requirement}

\section{Orbital Dynamics}

Few works analytically investigate orbital dynamics with atmospheric drag \cite{scheifele1977singularity,mittleman1982analytic,carter2002clohessy,xu2011analytical,cao2015linearized}. To describe the leader (subscript $l$) and follower (subscript $f$) relative motion, we employ the HCW plus drag formulation by Silva \cite{silva2008formulation}. Let $\delta\bm{r}^i=\bm{r}_l^i-\bm{r}_f^i$ and $\delta\bm{r}^h=\matrice{R}_i^h \delta\bm{r}^i=[x^h,y^h,z^h]^\top$ its projection on $\mathcal{F}_h$, the nonlinear relative dynamics can be approximated as 
\sottosistema{\ddot{x}^h&=-\beta_lS_lw_f\norm{\bm{\omega}_\oplus^i} \cos(\mathscr{i}_f)y^h+3\nu_f^2x^h+2\nu_f\dot{y}^h-\beta_lS_lw_f\dot{x}^h
\label{xhill}\\
   \ddot{y}^h&=-2\beta_lS_lw_f\norm{\bm{\omega}_\oplus^i} \cos(\mathscr{i}_f) x^h-2\beta_lS_lw_f \dot{y}^h-2\nu_f\dot{x}^h-(\beta_lS_l-\beta_fS_f)w_f^2
\label{yhill}\\
    \ddot{z}^h&=-\nu_f^2z^h-\beta_lS_lw_f \dot{z}^h
\label{zhilleq}}{\label{sistemahill}}

\noindent
where $\nu_f=(\mu/\norm{\bm{r}_f^i}^3)^{1/2}$, $w_f=\norm{\bm{w}_f^i}$, and $\mathscr{i}_f\in [0,\pi]$ is the follower's orbit inclination. 

From an astrodynamics perspective, the control inputs are \( S_l \) and \( S_f \). Because from \eqref{betacoeffient} $\beta=\beta(M,\rho, C_D)$ and \( S_l \) appears as a state multiplier, \eqref{sistemahill} is not an LTI system. The model can be further refined by incorporating additional orbital perturbations (\eg, the $J_2$ effect \cite{schweighart2002high}), and the OC can be achieved for a generic $\beta(t)$ with uncertainties using adaptive control  \cite{riano2021adaptive} or desensitized approaches \cite{harris2019desensitized}. However, since designing a precise OC is not the primary focus of this work, we will rely on the following simplifying assumptions, similar to those taken by Harris et al. \cite{harris2020linear}.

\begin{assumption} \label{Assumption3}
 $\beta=\beta_l\approx\beta_f\approx \bar{\beta} =const.$ within the operating altitude range.    
\end{assumption}

\begin{assumption} \label{Assumption4}
The leader is a virtual CubeSat with a constant ES, \ie, \( S_l(t) = \bar{S}_l \).
\end{assumption}

\refassumption{Assumption3} is motivated by three reasons. 
First, CubeSats in formation flight are typically identical ($M_l = M_f$). Second, a short control duration does not cause significant orbital decay to produce huge $\rho$-variations. Furthermore, refined atmospheric models like NRLMSISE-00, which well capture density variations, show that $\rho$ exhibits oscillatory behavior without abrupt variations over long-term propagations \cite{CALABIA2020105207}. Thus, in this context, the OC preliminary design can assume $\rho=\bar{\rho}$, where $\bar{\rho}$ denotes the mean atmospheric density computed over a representative long-term propagation. Finally, although $C_D$ depends on several parameters (\eg, the rarefaction regime, surface accommodation, and temperature ratios \cite{sun2017roto}), for CubeSats $C_D\in[2,2.4]$ \cite{moe2005gas} and its impact on the drag acceleration is much smaller than the uncertainty on $\rho$. Thus, we adopt the practical initial approximation $C_D\approx \bar{C}_D=2.2$.

Furthermore, we do not include the $J_2$ perturbation. The latter introduces long-term drift and short-period oscillations on an orbital timescale. For short control durations, secular effects are negligible, while short-period effects cause small oscillations around the inter-satellite distance setpoint.

Since $z^h=0$ is an asymptotically stable equilibrium point of  \eqref{zhilleq}, the equation can be removed. Let $\bar{\beta}=\beta(M,\bar{\rho},\bar{C}_D)$, $\lambda_\mathscr{i}=\bar{\beta}\bar{S}_lw_f\norm{\bm{\omega}_\oplus^i} \cos(\mathscr{i}_f)$, and define the reduced state $\bm{x}^h=[x^h,y^h,\dot{x}^h,\dot{y}^h]^\top$; the final LTI system is
\begin{align}
    \dot{\bm{x}}^h=\matrice{A}\bm{x}^h+ \bm{b}^h\Delta S
    \label{obitequationfin}
\end{align}

\noindent
where $\Delta S=\bar{S}_l-S_f$, and  $\matrice{A}$  and  $\bm{b}^h$ are defined as
\withtablespacing{0.7}{
\begin{align}
    \matrice{A}=\left[\begin{array}{cccc}
        0&0 & 1&0  \\
        0&0&0&1\\
         3\nu_f^2 &-\lambda_\mathscr{i}&-\bar{\beta}\bar{S}_lw_f  &2\nu_f \\
        -2 \lambda_\mathscr{i} & 0& -2 \nu_f  & -2 \bar{\beta}\bar{S}_lw_f
    \end{array}\right], \quad \bm{b}^h=-\bar{\beta} w_f^2\left[\begin{array}{c}
         0  \\
         0\\
         0\\
         1
    \end{array}\right]
\end{align}
}

\section{Control Problem Statement}
Let $\bm{x}^h_d=[x_d^h,y_d^h,0,0]^\top$ be the reference state, $\bm{x}_e^h=\bm{x}^h-\bm{x}_d^h$, and $\dot{\bm{x}}_e^h=\matrice{A}\bm{x}_e^h+\matrice{A}\bm{x}_d^h+\bm{b}^h\Delta S$. Since this is an LTI system, we design $\Delta S_d$ to stabilize it as the sum of a constant term $\overline{\Delta S}\colon \matrice{A}\bm{x}_d^h+\bm{b}^h\overline{\Delta S}=0$ and a state feedback. Moreover, since $\matrice{A}\bm{x}_d^h=[0,0,3 x_d^h\nu_f^2-y_d^h \lambda_\mathscr{i},-2 x_d^h\lambda_\mathscr{i}]^\top$ and $\bm{b}^h=[0,0,0,-\bar{\beta} w_f^2]^\top$, no $\overline{\Delta S}$ nullifies the sum. Thus, we must find an accuracy trade-off. The modified HCW equations \eqref{sistemahill} show that the drag acceleration appears as a pure forcing term in the $y$-equation only. Thus, to maximize the control authority and minimize the time, the $\bm{e}_\theta$-separation is preferable. To find such $\overline{\Delta S}$, we evaluate the steady-state solution $\bm{x}_{e,ss}^h$, \ie
\begin{align}
    \bm{x}_{e,ss}^h=-\bm{x}_d^h-\matrice{A}^{-1}\bm{b}^h\overline{\Delta S} =-\left[x_d^h+\bar{\beta}  w_f^2\overline{\Delta S}/(2 \lambda_\mathscr{i}),y_d^h+3 \bar{\beta}  \nu _f^2 w_f^2\overline{\Delta S}/(2 \lambda_\mathscr{i}^2),0,0\right]^\top
\end{align}

\noindent
and, thus, zero $\bm{e}_\theta$ steady-state error implies
\begin{align}
    \overline{\Delta S}=-2  \lambda_\mathscr{i}\epsilon_r/( \bar{\beta}   w_f^2) \rightarrow \bm{x}_{e,ss}^h=\left[\epsilon_r-x_d^h,0,0,0\right]^\top, \quad \epsilon_r=y_d^h \lambda_\mathscr{i}/(3 \nu _f^2)
\label{SDandxess1}
\end{align}

The magnitude of $\lambda_\mathscr{i}$ is tiny (often neglected \cite{harris2019desensitized,harris2020linear}) and, since $y_d$ must be sufficiently small for the HCW assumptions, $\epsilon_r \approx 0$. Thus, in practice,  $x^h\rightarrow 0$ and $\overline{\Delta S}\approx 0$ $\forall \bm{x}_d^h$.     By defining $\bm{X}_e^h=\bm{x}_e^h-\bm{x}_{e,ss}^h$, the model equations are
\sottosistema{ \dot{\bm{X}}_e^h&=\matrice{A}\bm{X}_e^h+\matrice{A}(\bm{x}_d^h+\bm{x}_{e,ss}^h)+\bm{b}^h\Delta S
    \label{orbitalerrordynzeroo}\\
   \Delta\dot{S}&=-\dot{S}_f=-(\bm{\psi}^{b})^{\top}\bm{\omega}^b-\phi \label{DSdotNew}\\
    \matrice{J}\dot{\bm{\omega}}^b&=\Smat(\matrice{J}\bm{\omega}^b)\bm{\omega}^b+\bm{u}^b
    \label{Eulerclassic}}{\label{sistemada controllare}}

\noindent
where  $\dot{S}_f$ is given by
\eqref{Sdotsop}. Moreover, \eqref{Eulerclassic} are the  Euler equations, where $\matrice{J}$ is the symmetric and positive definite (SPD) inertia tensor, $\bm{\omega}^b$ is the follower angular velocity, and $\bm{u}^b$ is the control torque. To simplify the notation, we do not use the subscript $f$ for the attitude dynamics to indicate that the quantity refers to the follower. 

The control problem is then: find $\bm{u}^b \colon \{ \bm{\omega}^b\rightarrow \bm{\omega}^b_d, \, \Delta S\rightarrow \Delta S_d, \, \bm{X}_e^h\rightarrow 0  \}$ 
 for $t\rightarrow +\infty$, where $\bm{\omega}^b_d$ is a suitable angular velocity profile that meets \refrequirement{RequirementI} and $\Delta S_d=\bar{S}_l-S_{f,d}$ with $S_{f,d}$ an admissible  follower ES profile as per \refrequirement{reqSd}.

\section{Controller design}
In this section, we design an angular-velocity-based controller for ES and inter-satellite distance regulation. Unlike traditional approaches, no predefined reference profiles are required. Instead, the controller is designed using integrator backstepping, a recursive technique for stabilizing high-order nonlinear systems. The method proceeds from the inner dynamics outward, treating each subsequent state variable as a virtual control to stabilize the preceding one, while compensating the induced error at the next design step. Therefore, our main proposition is:

\begin{proposition} \label{MainTheorem}
Assume $\bm{m}=\hat{\bm{m}} \, \forall t \geq \hat{t}$ as per \refrequirement{RequirementI}. Moreover, let $\bm{g}\in\R^4$ be the LQR gain that stabilizes the plant of \eqref{orbitalerrordynzeroo} and $\matrice{P}$ be the solution of the associated algebraic Riccati equation (ARE). Finally, let $\Delta S_d=-\bm{g}^\top\bm{X}_e^h +\overline{\Delta S}$, $\Delta S_e=\Delta S-\Delta S_d$, and $\bm{\omega}_d^b$ be the \refrequirement{RequirementI}-compliant solution of the consistent-underdetermined  algebraic linear system 
\begin{align}
    \Lambda-(\bm{\psi}^{b})^{\top}\bm{\omega}^b_d+k_s \Delta S_e=0
    \label{mainequationequalitycons}
\end{align}

\noindent
with $k_s>0$ a feedback gain, $\Lambda=-\phi+[(\bm{b}^{h})^\top\matrice{P}+\bm{g}^\top\matrice{A}_c]\bm{X}_e^h+\bm{g}^\top\bm{b}^h\Delta S_e$, and $\matrice{A}_c=\matrice{A}-\bm{b}^h\bm{g}^\top$. Then, defining $\bm{\omega}_e^b=\bm{\omega}^b-\bm{\omega}_d^b$,  the equilibrium point $(\bm{X}_e^h,\Delta S_e,\bm{\omega}_e^b)=(\bm{0},0,\bm{0})$ of the error dynamics $\{\dot{\bm{X}}_e^h,\Delta \dot{S}_e,\dot{\bm{\omega}}_e^b\}$ is asymptotically stable if the rotational dynamics \eqref{Eulerclassic} is in closed-loop with the  control law
\begin{align}
    \bm{u}^b&=-\Smat(\matrice{J}\bm{\omega}^b)\bm{\omega}^b+\matrice{J}\dot{\bm{\omega}}^b_d+\Delta S_e\bm{\psi}^{b}-\matrice{K}_\omega\bm{\omega}_e^b
    \label{ufinal}
\end{align}
with $\matrice{K}_\omega \belong{3}{3}$ an SPD gain matrix.
\end{proposition}

\begin{proof}
The proof follows the classical steps of the integrator backstepping technique --\textit{cf.} Khalil \cite{khalil2002control}.
\paragraph{Step 1: Design an ideal $\Delta S_d$ profile to stabilize \eqref{orbitalerrordynzeroo} and proof of asymptotic stability of the eq. point $\bm{X}_e^h=\bm{0}$}
\ \\
Let $\matrice{P}\belong{4}{4}$ be an SPD matrix and $V_1\colon \R^4 \to \R$ be a continuously differentiable Lyapunov function such that $V_1(\bm{0})=0$ and $V_1(\bm{X}_e^h)>0$ in $\R^4-\{\bm{0}\}$ defined as $V_1(\bm{X}_e^h)=(\bm{X}_e^{h})^{\top} \matrice{P} \bm{X}_e^h/2$. By taking its time derivative and using \eqref{orbitalerrordynzeroo}, we get
\begin{align}
\dot{V}_1&=\frac{1}{2}(\dot{\bm{X}}_e^{h})^{\top} \matrice{P} \bm{X}_e^h+\frac{1}{2}(\bm{X}_e^{h})^{\top} \matrice{P} \dot{\bm{X}}_e^h=\frac{1}{2}(\bm{X}_e^{h})^\top [\matrice{A}^\top \matrice{P}+\matrice{P}\matrice{A}] \bm{X}_e^h+(\bm{X}_e^{h})^{\top} \matrice{P}[\matrice{A}(\bm{x}_d^h+\bm{x}_{e,ss}^h)+\bm{b}^h\Delta S]
\end{align}

Since the pair $(\matrice{A},\bm{b}^{h})$ is controllable,  by choosing $\matrice{P}$ as the solution of the ARE $\matrice{A}^\top \matrice{P}+\matrice{P}\matrice{A}-\matrice{P}\bm{b}^h\mathfrak{R}^{-1}(\bm{b}^{h})^{\top} \matrice{P}+\matrice{Q}=\bm{0}$
and $\Delta S \equiv  \Delta S_d=-\mathfrak{R}^{-1}(\bm{b}^{h})^{\top}\matrice{P}\bm{X}_e^h+\overline{\Delta S}=-\bm{g}^\top\bm{X}_e^h +\overline{\Delta S}$ with $\overline{\Delta S}$ defined in \eqref{SDandxess1}, we end up with
\begin{align}
    \dot{V}_1(\bm{X}_e^h)=-\frac{1}{2} (\bm{X}_e^h)^\top[\matrice{P}\bm{b}^h\mathfrak{R}^{-1}(\bm{b}^{h})^{\top} \matrice{P}+\matrice{Q}]\bm{X}_e^h=-(\bm{X}_e^h)^\top\matrice{\mathfrak{Q}}\bm{X}_e^h
\end{align}

\noindent
where $\matrice{Q}$ denotes the state cost matrix and $\mathfrak{R}$ the control cost. By picking $\matrice{Q}$ SPD (\textit{e.g.}, with the Bryson rule \cite{bryson2018applied,okyere2019lqr}) and, since $\matrice{P}\bm{b}^h\mathfrak{R}^{-1}(\bm{b}^{h})^{\top} \matrice{P}$ is symmetric and positive semi-definite,  the sum $\matrice{\mathfrak{Q}}$ is SPD. Since $\dot{V}_1(\bm{X}_e^h)<0$, $V_1(\bm{0})=0$, and $V_1(\bm{X}_e^h)>0 \, \forall \bm{X}_e^h\neq \bm{0}$, the origin of \eqref{orbitalerrordynzeroo} is asymptotically stable.

\paragraph{Step 2: Design of  $\bm{\omega}_d^b$ profile to stabilize \eqref{DSdotNew}, and proof of asymptotic stability of the eq. point $(\bm{X}_e^h,\Delta S_e)=(\bm{0},0)$ }
\
\\
 In the first backstep, we set $\Delta S_e=\Delta S-\Delta S_d$. By solving  for $\Delta S$ and inserting it into \eqref{orbitalerrordynzeroo}, we get $\dot{\bm{X}}_e^h=\matrice{A}_c\bm{X}_e^h+  \bm{b}^h\Delta S_e$,
where $\matrice{A}_c=\matrice{A}-\bm{b}^h\bm{g}^\top$ is the Hurwitz orbital closed-loop system matrix. Using this result and \eqref{DSdotNew}, $\Delta \dot{S}_e$ becomes
\begin{align}
    \Delta \dot{S}_e=-(\bm{\psi}^{b})^{\top}\bm{\omega}^b-\phi+\bm{g}^\top\dot{\bm{X}}_e^h=-(\bm{\psi}^{b})^{\top}\bm{\omega}^b-\phi+\bm{g}^\top\matrice{A}_c\bm{X}_e^h+\bm{g}^\top\bm{b}^h\Delta S_e
    \label{DSdotintermediate}
\end{align}

\noindent
In \eqref{DSdotintermediate}, we sum quantities expressed in $\mathcal{F}_b$ and $\mathcal{F}_h$ because scalar products are rotation invariants.

 Let $V_2\colon \R^4 \times \R \to \R$ be a continuously differentiable Lyapunov function such that $V_2(\bm{0},0)=0$ and $V_2(\bm{X}_e^h,\Delta S_e)>0$ in $\R^4\times \R-\{\bm{0},0\}$ defined as $V_2(\bm{X}_e^h,\Delta S_e)=(\bm{X}_e^{h})^{\top}\matrice{P}\bm{X}_e^h/2+\Delta S_e^2/2 $. By taking the time derivative, we get
 \begin{align}
      \dot{V}_2=-(\bm{X}_e^{h})^{\top}\matrice{\mathfrak{Q}}\bm{X}_e^h +\Delta S_e[\Lambda-(\bm{\psi}^{b})^{\top}\bm{\omega}^b]
 \end{align}

 \noindent
with $\Lambda=-\phi+[(\bm{b}^{h})^{\top}\matrice{P}+\bm{g}^\top\matrice{A}_c]\bm{X}_e^h+\bm{g}^\top\bm{b}^h\Delta S_e$. Thus, by choosing $\bm{\omega}^b\equiv \bm{\omega}_d^b$ such that
\begin{align}
 \Lambda-(\bm{\psi}^{b})^{\top}\bm{\omega}^b_d=-k_s \Delta S_e
\end{align}

\noindent
with $k_s>0$, we end up with $\dot{V}_2(\bm{X}_e^h,\Delta S_e)=-(\bm{X}_e^{h})^{\top}\matrice{\mathfrak{Q}}\bm{X}_e^h-k_s\Delta S_e^2<0$. Since $\dot{V}_2(\bm{X}_e^h,\Delta S_e)<0$, $V_2(\bm{0},0)=0$, and $V_2(\bm{X}_e^h,\Delta S_e)>0\, \forall (\bm{X}_e^h,\Delta S_e)\neq (\bm{0},0)$, the origin of the system \eqref{orbitalerrordynzeroo} and \eqref{DSdotintermediate} is asymptotically stable.

\paragraph{Step 3:  Design of $\bm{u}^b$ and proof of asymptotic stability of the eq. point $(\bm{X}_e^h,\Delta S_e,\bm{\omega}_e^b)=(\bm{0},0,\bm{0})$} 
\
\\
The last back-step is otained by setting $\bm{\omega}_e^b=\bm{\omega}^b-\bm{\omega}_d^b$ and, taking its time derivative and using \eqref{Eulerclassic}, we get 
\begin{align}
\matrice{J}\dot{\bm{\omega}}_e^b=\Smat(\matrice{J}\bm{\omega}^b)\bm{\omega}^b-\matrice{J}\dot{\bm{\omega}}^b_d+\bm{u}^b
\label{omeerror}
\end{align}

Let $V_c\colon \R^4\times\R\times\R^3 \to \R$ be a continuously differentiable composite Lyapunov function such that $V_c(\bm{0},0,\bm{0})=0$ and $V_c(\bm{X}_e^h,\Delta S_e,\bm{\omega}_e^b)>0$ in $\R^4\times\R\times\R^3-\{\bm{0},0,\bm{0}\}$ defined as $ V_c(\bm{X}_e^h,\Delta S_e,\bm{\omega}_e^b)=(\bm{X}_e^{h})^{\top}\matrice{P}\bm{X}_e^h/2+\Delta S_e^2/2+(\bm{\omega}_e^{b})^{\top}\matrice{J}\bm{\omega}_e^b/2$. By taking the time derivative of such a  composite Lyapunov function, we get
\begin{align}
    \dot{V}_c =-(\bm{X}_e^{h})^{\top}\matrice{\mathfrak{Q}}\bm{X}_e^h-k_s\Delta S_e^2+(\bm{\omega}_e^{b})^{\top}\left[\Smat(\matrice{J}\bm{\omega}^b)\bm{\omega}^b-\matrice{J}\dot{\bm{\omega}}^b_d-\Delta S_e\bm{\psi}^{b}+\bm{u}^b\right]
\end{align}

\noindent
and, by defining the control input $\bm{u}
^b$ as
\begin{align}
    \bm{u}^b=-\Smat(\matrice{J}\bm{\omega}^b)\bm{\omega}^b+\matrice{J}\dot{\bm{\omega}}^b_d+\Delta S_e\bm{\psi}^{b}-\matrice{K}_\omega\bm{\omega}_e^b
\end{align}

\noindent
with $\matrice{K}_\omega$ an SPD gain matrix, we end up with $\dot{V}_c(\bm{X}_e^h,\Delta S_e,\bm{\omega}_e^b)=-(\bm{X}_e^{h})^{\top}\matrice{\mathfrak{Q}}\bm{X}_e^h-k_s\Delta S_e^2-(\bm{\omega}_e^{b})^{\top}\matrice{K}_\omega\bm{\omega}_e^b<0$. Since $\dot{V}_c(\bm{X}_e^h,\Delta S_e,\bm{\omega}_e^b)<0$, $V_c(\bm{0},0,\bm{0})=0$, and $V_c(\bm{X}_e^h,\Delta S_e,\bm{\omega}_e^b)>0\, \forall (\bm{X}_e^h,\Delta S_e,\bm{\omega}_e^b)\neq (\bm{0},0,\bm{0})$, the origin of the system \eqref{orbitalerrordynzeroo}, \eqref{DSdotintermediate}, and \eqref{omeerror} is asymptotically stable.

\end{proof}

\section{Minimization Problem}
The main challenge with the controller designed in \refproposition{MainTheorem} is solving the consistent, underdetermined algebraic linear system \eqref{mainequationequalitycons} to determine the \refrequirement{RequirementI}-compliant $\bm{\omega}_d^b$ profile. Specifically, \eqref{mainequationequalitycons} is one equation with three unknowns. Thus, by the Rouché-Capelli theorem \cite{shafarevich2012linear}, it has infinitely many solutions in a two-dimensional solution space.

For this reason, an effective resolution strategy is to set up a constrained minimization problem. To use this approach, a well-defined set of constraints is crucial. The criticality of the configuration constraints \eqref{conficonstr} is that, at each time step, $\matrice{R}_i^b(\bm{\omega}^b_d)$ is not available without an exact solution to \eqref{kinematicequation}. To tackle this issue, our control algorithm estimates $\matrice{R}_i^b(\bm{\omega}^b_d)$ with a precise approximation of the unit quaternion $\bm{q}(\bm{\omega}_d^b,t)=[q_1,q_2,q_3,q_4]^\top=[\bm{q}_{1:3}^\top,q_4]^\top$  \cite{wertz2012spacecraft}, \ie,
\begin{align}
   \bm{q}(t+\tau)&=\left[\cos\left(\norm{\bm{\omega}_d^b} \tau/2 \right)\matrice{I}_4+ \norm{\bm{\omega}_d^b}^{-1}\sin\left(\norm{\bm{\omega}_d^b} \tau/2\right)\matrice{\Theta}(\bm{\omega}_d^b)\right]\bm{q}(t)  \\
   \matrice{R}_i^b(\bm{q})&=(q_4^2-||\bm{q}_{1:3}||^2)\matrice{I}_3+2\bm{q}_{1:3}\bm{q}_{1:3}^\top-2q_4\Smat(\bm{q}_{1:3})
    \label{quat2R}
\end{align}
\noindent
where $\tau$ is the time step. In other words,
our algorithm replaces  \eqref{conficonstr} with $\diag(\hat{\bm{m}})\matrice{R}_i^b(\bm{q}(t+\tau))\bm{\xi}^i\geq 0$. Thus, the constrained minimization problem formalizes into 
\withtablespacing{\spaziatura}{
\begin{align}
    \begin{aligned}
      \bm{\omega}_d^b=  \min_{\bm{\omega}_d^b\in\R^3} &\quad \frac{1}{2} (\bm{\omega}_d^b)^\top \bm{\omega}_d^b\\
        \rm{s.t.}&     \quad   \bm{\Gamma}_\omega
    \end{aligned} , \quad \bm{\Gamma}_\omega=\left[\begin{array}{c}
         \Lambda-(\bm{\psi}^{b})^{\top}\bm{\omega}^b_d+k_s \Delta S_e=0 \\
         \diag(\hat{\bm{m}})\matrice{R}_i^b(\bm{\omega}^b_d)\bm{\xi}^i\geq 0 \\
         \rm{additional}\, \mathrm{constraints}
    \end{array}\right]
    \label{minprobbello}
\end{align}
}
\noindent
where $\bm{\Gamma}_\omega$ is the constraints set that includes \eqref{mainequationequalitycons} as a linear equality constraint and \eqref{conficonstr} as a nonlinear inequality.

\section{Simulations and Results}
This section assesses the performance of the control algorithm designed in \refproposition{MainTheorem} through several simulated cases. The nominal initial conditions and simulation parameters are detailed in \reftable{orbtab}, \reftable{HCWtab}, and \reftable{simpartab}.

\withtablespacing{\spaziatura}{
\begin{table}[hbt!]
\caption{ Orbital initial conditions}
\label{orbtab}
\centering
\begin{tabular}{lcccccccc}
\toprule

Spacecraft& $\mathscr{a}$ $[$km$]$ & $\mathscr{e}$ $[-]$ & $\mathscr{i}$ $[$deg$]$ & $\mathscr{\Omega}$ $[$deg$]$ & $\mathscr{\varpi}$ $[$deg$]$ & $\theta$ $[$deg$]$ & $T$ $[$h$]$ & $\rho$ $[$kg/m$^3]$\\
\midrule
Follower & $6821.000$ & $1\cdot10^{-6}$ & $90.0000000$& $0$ &$0$ &$0$ & $1.5573$ & $3.0137\cdot 10^{-12}$\\
Leader &   $6821.002$ & $2\cdot10^{-6}$ &$89.9999989$&$1\cdot10^{-6}$&$1.2\cdot10^{-6}$&$1\cdot10^{-6}$& $1.5573$&$3.0137\cdot 10^{-12}$\\
\bottomrule
\end{tabular}
\end{table}
}

\withtablespacing{\spaziatura}{
\begin{table}[hbt!]
\caption{ HCW, attitude, and ES initial conditions}
\label{HCWtab}
\centering
\begin{tabular}{lll}
\hline
\toprule
Quantity& Value &Units \\
\midrule
$\delta \bm{r}^h$ & $[-4.821,0.26191,-0.11905]^\top$& m\\
$\delta \dot{\bm{r}}^h$ & $[-1.9367\cdot10^{-10},     0.011927,  -0.00014676]^\top$& m/s\\
$\norm{\delta \bm{r}^h}$ & $4.8296$& m\\

$\bm{\omega}^b$ & $[0,0,0]^\top$ & rad/s\\
$\bm{q}$ & $[0.3774,       0.2877,       0.6255,       0.6193]^\top$ & $-$\\
$\bm{\eta}^b$ & $[ 0.05114, 0.83,  0.55541]^\top$ &$-$\\
$\hat{\bm{m}}$ & $[1,1,1]^\top$ & $-$\\
$S_f$ & $1765.2638$ & cm$^2$\\
$\Delta S$ & $-165.26$ & cm$^2$\\
$\bar{\rho}$&$1.6611\tenE{-12}$& kg/m$^3$\\
\bottomrule
\end{tabular}
\end{table}
}

\withtablespacing{\spaziatura}{
\begin{table}[hbt!]
\caption{ Simulation parameters }
\label{simpartab}
\centering
\begin{tabular}{lll}
\toprule

Quantity& Value &Units \\
\midrule
$M$ & 1 & kg\\
$\matrice{J}$&$\diag([0.9516,0.9203,0.0527])$& kg$\cdot$m$^2$\\
$\bm{s}$ & $[748.3,1246.6,1246.6]^\top$ &cm$^2$\\
$\Gamma _s$ &$[748.2954,1915.1378]$& cm$^2$\\
$\bar{S}_l$ & $1600$ & cm$^2$\\
$\hat{S}_l$ & $100,200,300,400$ & cm$^2$\\
$\bm{x}^h_d$ & $[0,60,0,0]^\top$ &m,m,m/s,m/s\\
$k_s$ & $0.1$ & $1/$s\\
$\matrice{K}_\omega$ & $\matrice{I}_3$ &N$\cdot$m$\cdot$s\\
$\tau$ & $0.1$ & s\\
$\bm{l}^i$ & $[1,0,0]^\top$ & $-$\\
$\alpha_{\max}$ & $100$ & deg\\ 
\bottomrule
\end{tabular}
\end{table}
}

Both spacecraft are initially orbiting on similar $450$ km-altitude quasi-circular polar orbits of period $T=1.5573$ h defined by the orbital elements in \reftable{orbtab}. In our notation, $\mathscr{a}$ is the semimajor-axis, $\mathscr{e}$ the eccentricity, $\mathscr{i}$ the inclination, $\Omega$ the right ascension of the ascending node, $\varpi$ the argument of perigee, and $\theta$ the true anomaly. The Cartesian state conversion is omitted, as it follows the standard procedure described in \cite{curtis2020orbital}. Simulations include:
\begin{enumerate}
    \item the nonlinear dynamics \eqref{2body} with the gravitational acceleration, the drag acceleration, and the $J_2$ perturbation
    \item the density as a function of the geodetic latitude, longitude, and altitude, computed with the NRLMSISE-00 atmospheric model with daily solar flux $F_{10.7}$, 81-day averaged solar flux $\overline{F}_{10.7}$, and geomagnetic index $A_p$, obtained from the OMNI database and NOAA space weather data 
    \item an uncertain leader ES profile $S_l(t)=\bar{S}_l+\hat{S}_l\varepsilon_{rand}(t)$, where $\hat{S}_l$ is the amplitude of the uniformly distributed random noise $\varepsilon_{rand}$ caused by a spurious tumbling
    \item a multiplicative uncertainty on the  $\beta$ coefficient, modeled as $\beta = \bar{\beta}(1+\varepsilon_\beta)$, with $\varepsilon_\beta \in [-20\%,\,20\%]$, to capture the uncertain variability of $\rho$ and  $C_D$, and to check the robustness of the control system
    \item the configuration-belonging constraint \eqref{conficonstr}
    \item an additional attitude constraint (denoted as $\alpha$-constraint): the angle $\alpha$ between $\bm{x}^b$ ($x$-axis of $\mathcal{F}_b$) and the inertial direction $\bm{l}^i$ must be always less than a maximum value
   $\alpha_{max}$, \ie, $(\bm{x}^b)^\top \matrice{R}_i^b(\bm{\omega}_d^b)\bm{l}^i> \cos(\alpha_{\max})$. 
\end{enumerate}

Four case studies were considered. The first without $J_2$ and without the $\alpha$-constraint, the second with $J_2$ and without the $\alpha$-constraint, the third without $J_2$ and with the $\alpha$-constraint, and the fourth with $J_2$ and the $\alpha$-constraint.

By denoting with $\mathscr{d}=\norm{\delta \bm{r}^h}$ the inter-satellite distance, the desired setpoint is $\mathscr{d}_d=60$ m. Control is initiated after detumbling with $\bm{\omega}^b(0)=\bm{0}$ rad/s and $\matrice{R}_i^b(0)$ computed via \eqref{quat2R} and the initial quaternion from \reftable{HCWtab}. Moreover,
$\hat{\bm{m}}=[1,1,1]^\top$, meaning that the follower is in configuration I of \reftable{tabsurfpos}. The LQR design provided
\withtablespacing{0.7}{
\begin{align}
    \bm{g}=
    \left[\begin{array}{c}
          -24.135\\      0.94868  \\    -3395.6   \\    -11181
    \end{array}
    \right]\tenE{-6},\quad  \matrice{P}=\left[\begin{array}{cccc}
0.050083&   -0.0023172   &    7.7172   &    22.643\\
 -0.0023172 &  0.00030127 &    -0.60764  &   -0.89007\\
  7.7172  &   -0.60764   &    2352.8 &      3185.8\\
    22.643  &   -0.89007   &    3185.8     &   10490
    \end{array}\right]\tenE{-5}
\end{align}
}

Simulations were performed over $10T$ with a time step of $\tau=0.1$ s, using Simulink and a $4$th-order Runge-Kutta integration scheme. The constrained minimization problem \eqref{minprobbello} was solved using the \textit{fmincon} solver with a sequential quadratic programming algorithm. Initial guesses were chosen as small perturbations of the current $\bm{\omega}^b$. To prevent large discontinuities in $\bm{\omega}^b$, an upper bound was imposed on its maximum component values.

\begin{figure}[H]
    \centering
    
\includegraphics[width=1\linewidth,trim=0cm 0cm 9cm 0cm,clip]{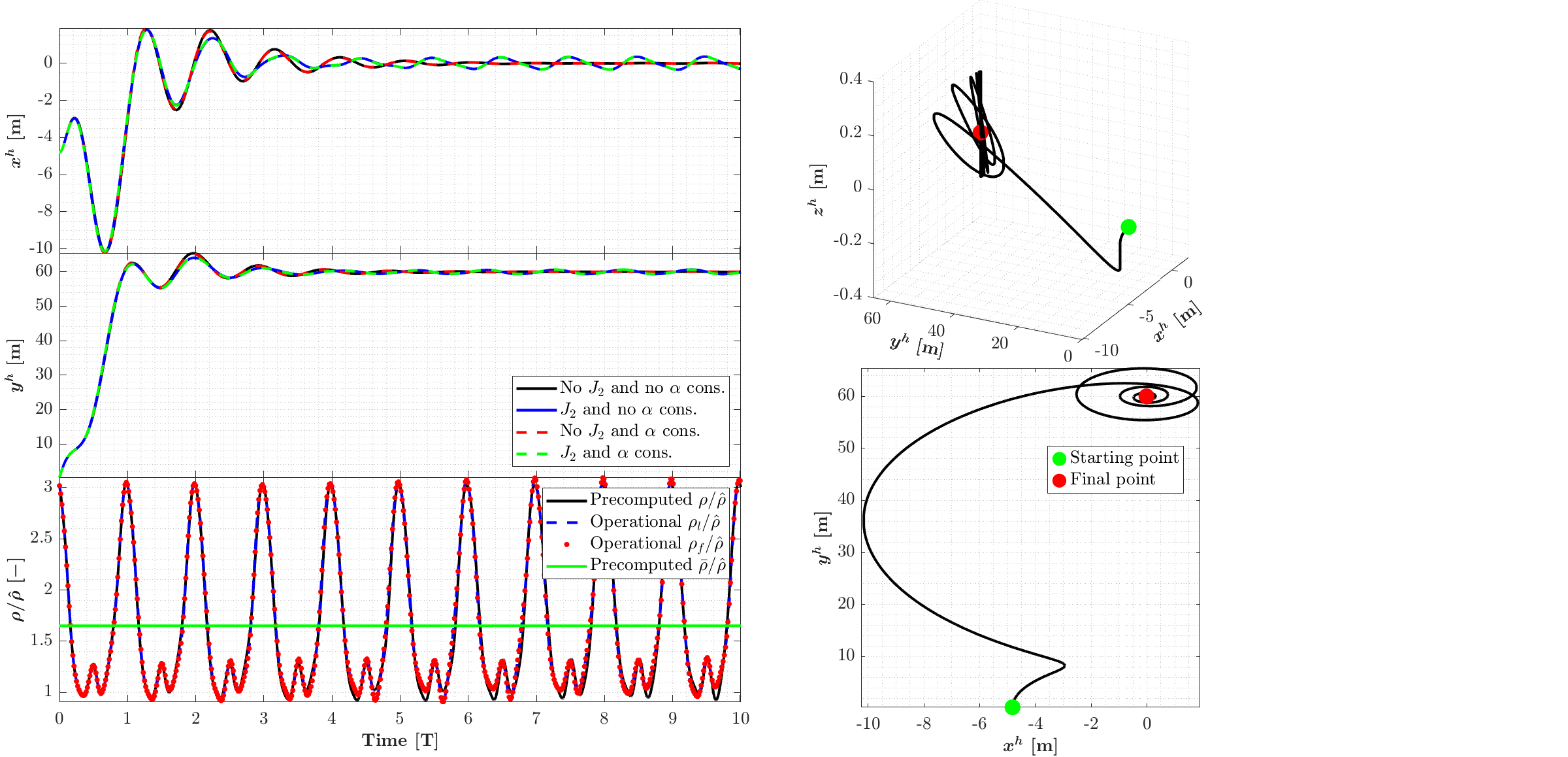}
    \caption{Hill coordinates, atmospheric density ($\hat{\rho}=10^{-12}$ kg/m$^3$), and relative motion in $\mathcal{F}_h$.}
    \label{figura3}
\end{figure}

Figure~\ref{figura3} shows the leader-follower relative motion in $\mathcal{F}_h$. Across all four scenarios, the system achieves convergence to the desired set point in the $x^h-y^h$ plane within approximately $6T$. The inclusion of the $J_2$ acceleration results in oscillations around the equilibrium, in line with the expected short-period effects of the perturbation. Secular effects are absent because of the short duration of the control operations. The 3D trajectory displays a minor motion in the $z^h$ direction, primarily due to three factors: imperfect confinement of $\bm{w}^i$ to the orbital plane, inclination differences, and the $J_2$ effect. Moreover, the density plot shows that the leader/follower operational density closely matches the precomputed profile. The observed oscillations are consistent with the NRLMSISE-00 model \cite{CALABIA2020105207}, making the use of the mean value  $\bar{\rho}=1.6611\tenE{-12}$ kg/m$^3$ an appropriate assumption for preliminary design of the OC via the HCW model \eqref{sistemahill}.

\begin{figure}[H]
    \centering
    
\includegraphics[width=1\linewidth,trim=3cm 1.5cm 2cm 1cm, clip]{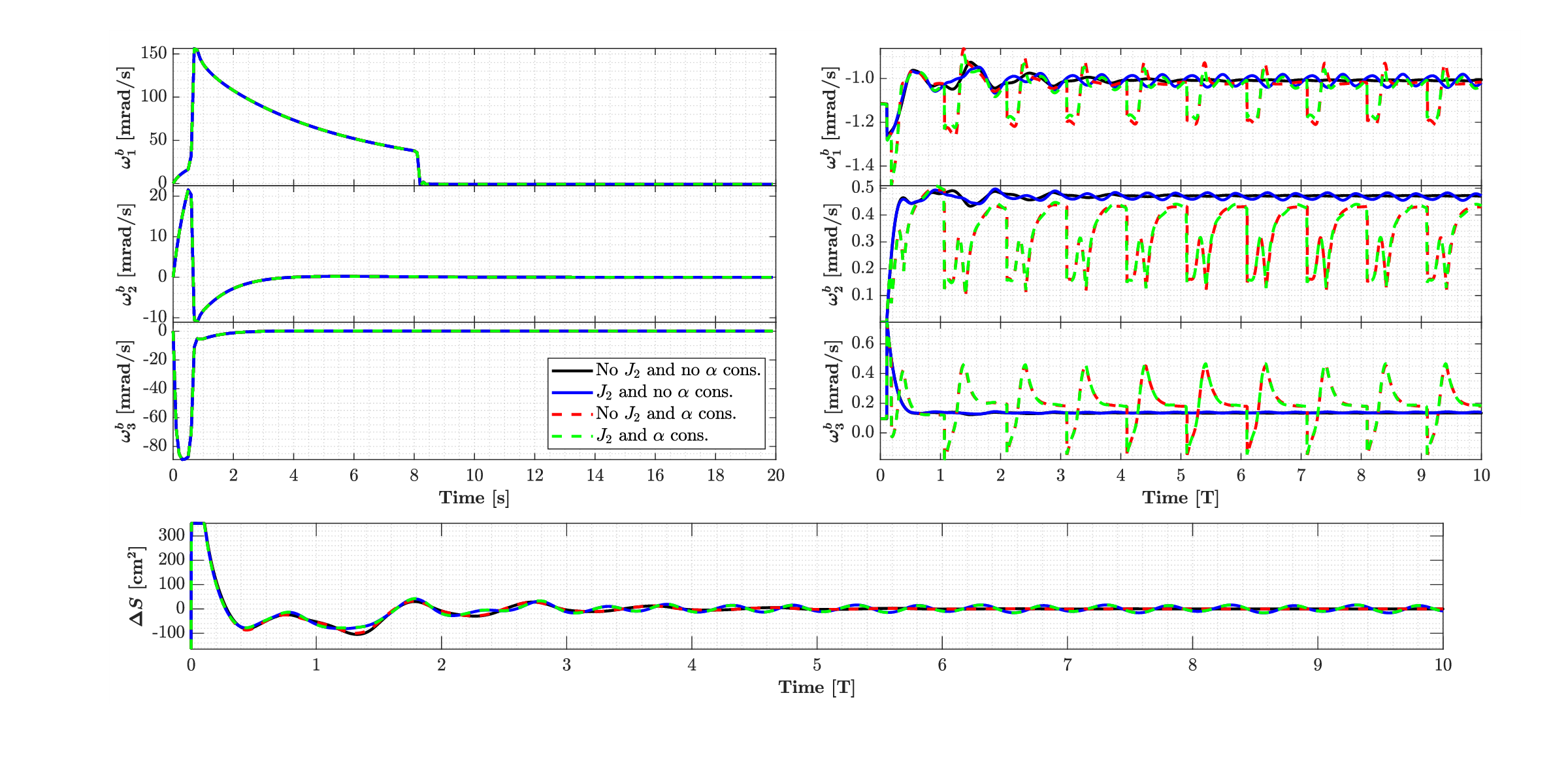}
    \caption{Angular velocity and ES difference.}
    \label{figura4}
\end{figure}

\begin{figure}[H]
    \centering
\includegraphics[width=1\linewidth,trim=0.1cm 0cm 2cm 0.5cm, clip]{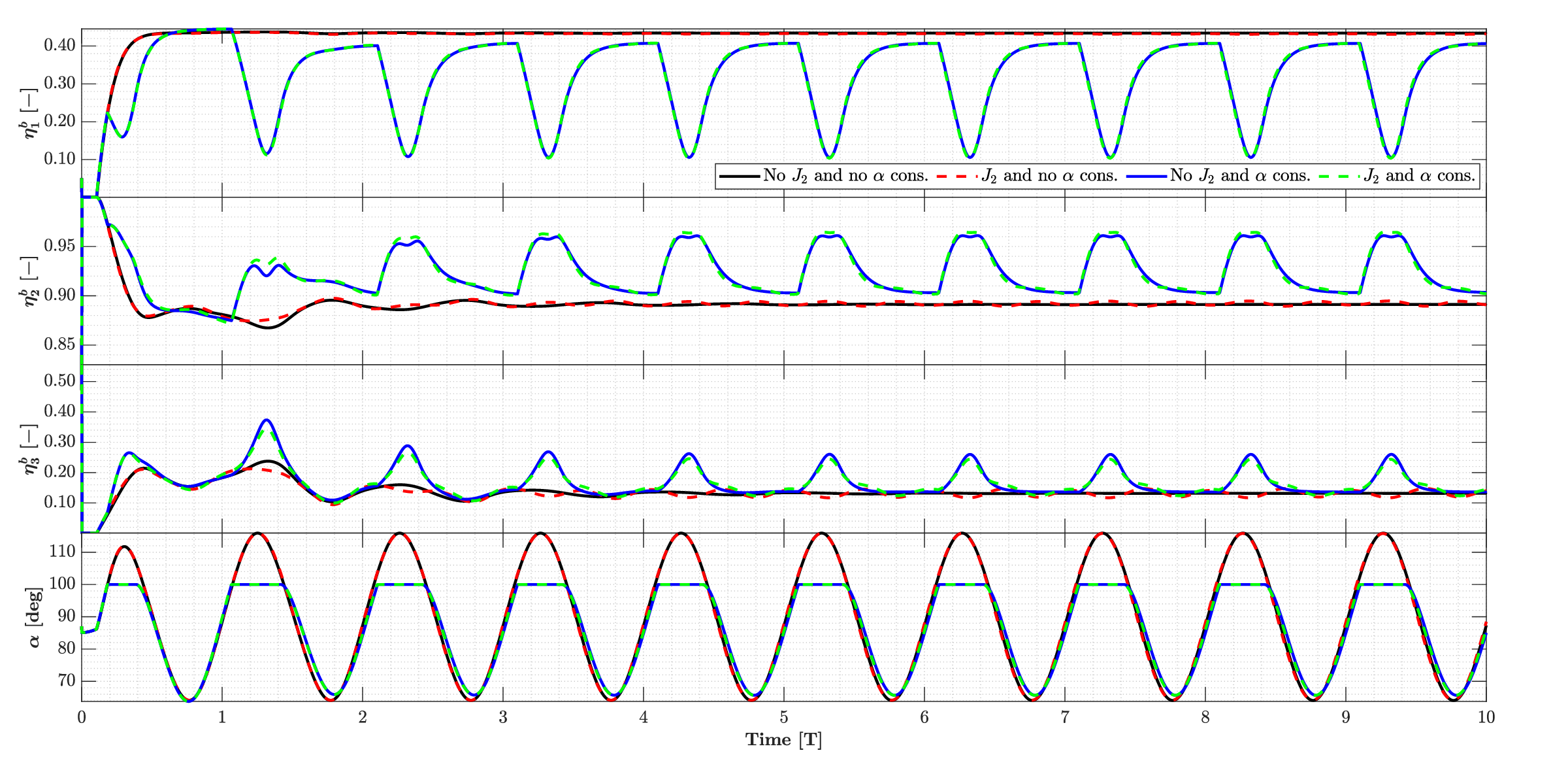}
    \caption{$\bm{\eta}^b$ components and $\alpha$-constraint.}
    \label{figura5}
\end{figure}

The required $\Delta S$ behaviors are shown in \reffig{figura4}. Such profiles imply that $S_{f}\in[1248, 1765.3]$ cm$^2 \subset \Gamma_s$, which satisfies \refrequirement{reqSd}. The $\bm{\omega}^b$-components are shown in \reffig{figura4}. These temporal evolutions ensure satisfaction of the configuration constraint \( \bm{\eta}^b > \bm{0} \) for all \( t \geq 0 \) as shown in \reffig{figura5}. After an initial transient during which the profiles of all four case studies overlap, two distinct behaviors emerge. In the cases without the $\alpha$-constraint, the $\bm{\omega}^b$-components stabilize at dynamic set points that resemble low-amplitude sinusoidal oscillations. In contrast, when the $\alpha$-constraint is enforced, the dynamic set points exhibit more irregular patterns, with larger and varying amplitudes. This adaptation ensures optimal compliance with the additional attitude constraint illustrated in \reffig{figura5}.  

\begin{figure}[H]
    \centering
   
\includegraphics[width=1\linewidth,trim=0cm 0.1cm 8cm 0.05cm,clip]{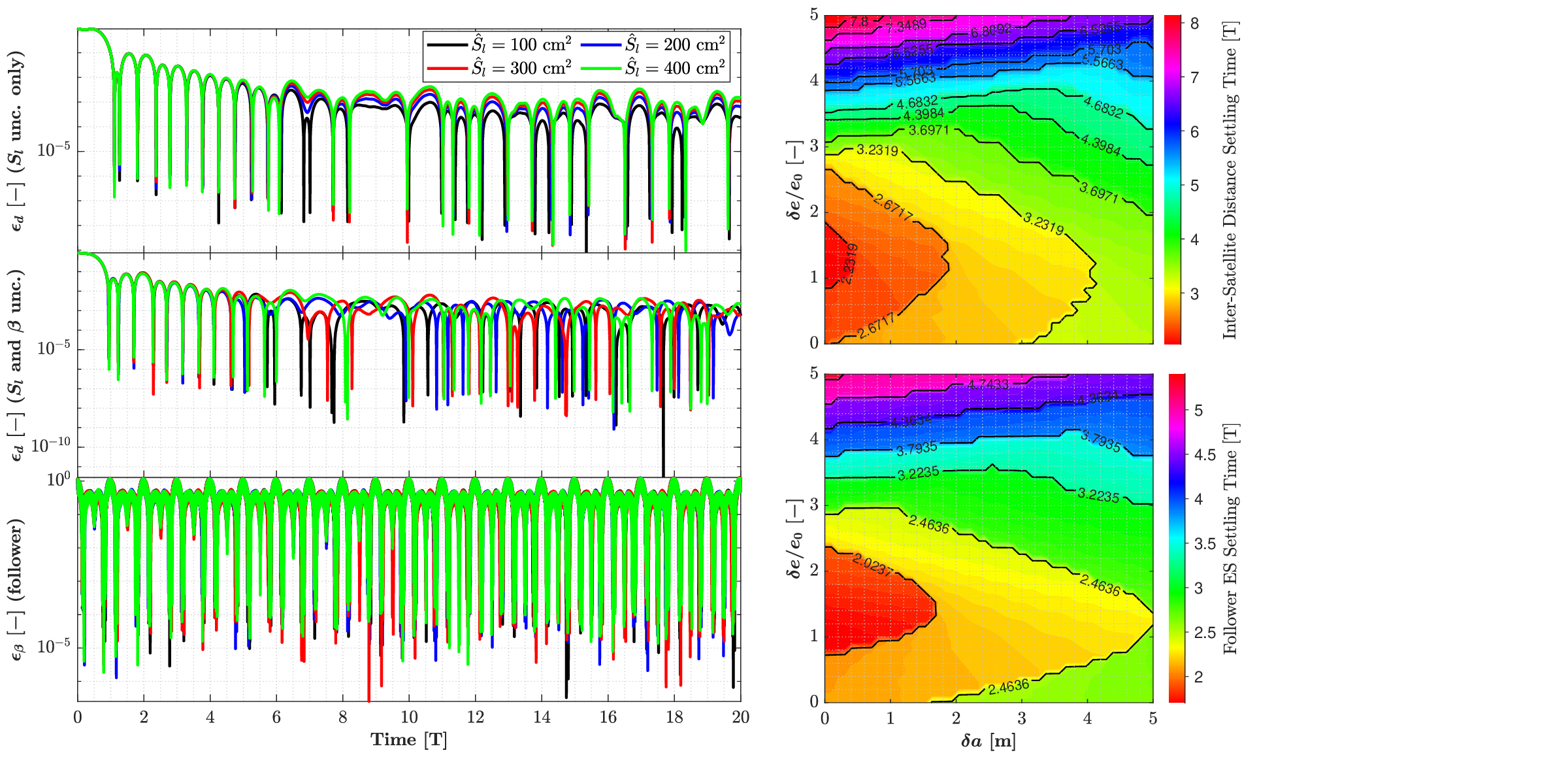}
    \caption{Sensitivity analysis on the inter-satellite distance, $\beta$-coefficient, and orbital initial conditions.}
    \label{figura6}
\end{figure}

We conducted a sensitivity analysis to assess the impact of \refassumption{Assumption3} and \refassumption{Assumption4} on the primary performance metric, \ie, the inter-satellite distance $\mathscr{d}$. Two sets of simulations were performed without $J_2$ and the $\alpha$-constraint. We increased the simulation duration to $20T$ to better see the accumulation of error over time. To analyze the performances meaningfully, we defined two relative errors: $\epsilon_\mathscr{d}=\lvert \mathscr{d}-\mathscr{d}_d\rvert/\mathscr{d}_d$ and $\epsilon_\beta=\lvert\beta-\bar{\beta}\rvert/\bar{\beta}$. The results are shown in \reffig{figura6} (left) for increasing $\hat{S}_l$-values on a logarithmic scale. The first simulation set, which includes only the uncertain $S_l$, shows that random noise minimally affects the transient phase ($t<6T$), where peak magnitudes decrease, and primarily introduces steady-state oscillations. In the worst-case scenario ($\hat{S}_l=400$ cm$^2$), the maximum $\epsilon_\mathscr{d}$-value remains low at approximately $0.72\%$, equating to a discrepancy of $43.27$ cm from the target $60$ m. The second set of simulations exhibits similar behavior but a higher oscillation frequency at steady state. Because of the uncertain $\beta$ coefficient. in the worst-case scenario, the maximum $\epsilon_\mathscr{d}$ is about $1.12\%$, corresponding to a $55.5\%$ increase with respect to the previous set of simulations. Nevertheless, this uncertainty analysis shows that the control system can still achieve an effective separation close to the desired one, thereby providing ultimate boundedness in the presence of uncertainties.

Finally, we conducted a parametric simulation campaign to assess the sensitivity of the OC part of the algorithm to variations in the initial orbital conditions. The initial leader and follower orbital elements are those in \reftable{orbtab},
except for $\mathscr{a}$ and $\mathscr{e}$ where we introduced the offsets $\delta \mathscr{a} \in [0,5]$ m and $\delta \mathscr{e} \in [0,5] \mathscr{e}_f$. Two performance metrics, denoted with $\lambda$, were considered: the settling times of the inter-satellite distance and of the follower ES, defined through the condition
$\lvert \lambda(t)-\lambda_{final} \rvert \leq 0.02 \lvert \lambda_{final}-\lambda_{initial} \rvert$. Moreover, we imposed the constraint of full usage of the admissible range $\Gamma_s$. The results are displayed in \reffig{figura6} (right). The primary outcome is that, in both cases, increasing $\delta \mathscr{a}$ and $\delta\mathscr{ e}$ does not lead to a monotonic increase in the settling time. Instead, two nearly parabolic regions emerge, within which the absolute minima of the performance metrics are attained. By comparing these regions, a global minimum can be identified. This result is particularly useful for preliminary mission design. Indeed, within the virtual leader framework, selecting initial conditions close to the optimal region reduces overall execution time.

\section{Conclusion}
This work presents an efficient control algorithm for CubeSat formation-flight relative positioning via ES management. This research activity draws two main conclusions. The first is that a coupled-orbit-attitude controller can be designed to regulate the inter-satellite distance without reference attitude precomputations, using angular velocity as the primary control variable. In this way, the attitude becomes merely a derived quantity. The second is that relative positioning control can be performed while satisfying additional attitude constraints. Finally, under the specified configuration and orbital assumptions, this work rigorously proves the asymptotic stability of the equilibrium points of the orbit-attitude coupled system in closed loop with the designed controller.

\section*{Acknowledgment}
This work was funded by the Research Council of Norway through the project 335832 QBDebris: A CubeSat formation for space debris characterization. The authors declare that AI tools were used solely for English grammar checking and improvement of written language. No theoretical or numerical technical content was generated using AI.

\bibliography{sample}

\end{document}